\documentclass[%
 reprint,
superscriptaddress,
 amsmath,amssymb,
pra,
]{revtex4-2}

\usepackage{amsmath,amssymb,amsfonts}
\usepackage{amsthm}
\usepackage{graphicx}
\usepackage{xcolor}
\usepackage{bm}
\usepackage{hyperref}
\usepackage{tikz}
\usetikzlibrary{arrows.meta,positioning,fit,backgrounds,calc,decorations.pathreplacing}
\pgfdeclarelayer{background}\pgfsetlayers{background,main}
\usepackage{enumitem}
\hypersetup{colorlinks=true, linkcolor=blue, citecolor=red, urlcolor=blue}
\usepackage{graphicx}
\usepackage{dcolumn}
\usepackage{bm}
\usepackage{hyperref}
\usepackage{amsthm}
\usepackage{booktabs}
\usepackage{xcolor}
\usepackage{orcidlink}

\newcommand{\BDCZ}{\mathrm{BDCZ}}
\newcommand{\DAHR}{\mathrm{DAHR}}
\newcommand{\FP}{F_{P}}
\newcommand{\FS}{F_{S}}
\newcommand{\Fpz}{F_{P}^{(0)}}
\newcommand{\Fsz}{F_{S}^{(0)}}
\newcommand{\Fstar}{F_{*}}
\newcommand{\eBSA}{\eta_{\rm BSA}}
\newcommand{\emeas}{\eta_{m}}
\newcommand{\pg}{p_{g}}
\newcommand{\HE}{{\rm HE}}
\newcommand{\Phip}{\Phi^{+}}
\newcommand{\Phim}{\Phi^{-}}
\newcommand{\Psip}{\Psi^{+}}
\newcommand{\Psim}{\Psi^{-}}
\newcommand{\ket}[1]{\lvert#1\rangle}
\newcommand{\bra}[1]{\langle#1\rvert}
\newcommand{\ketbra}[2]{\lvert#1\rangle\langle#2\rvert}
\newcommand{\wzero}{w_{0}}
\newcommand{\wone}{w_{1}}

\newtheorem{theorem}{Theorem}
\newtheorem{lemma}{Lemma}
\newtheorem{corollary}{Corollary}
\newtheorem{remark}{Remark}

\begin{document}

\title{Photon-efficient quantum repeater chains via hyperentanglement-assisted purification}

\author{Vinay Kumar\,\orcidlink{0000-0002-4635-3237}}\email{vk.quantph@gmail.com}
\affiliation{Department of Information Engineering, University of Pisa, Pisa, Italy}\affiliation{Institute for Informatics and Telematics (IIT), National Research Council (CNR), Pisa, Italy}

\author{Krishan Joshi\,\orcidlink{0009-0004-7114-2476}}
\affiliation{Department of Physics, University of Naples Federico II, Naples, Italy}\affiliation{National Institute of Optics (INO), National Research Council (CNR), Florence, Italy}

\author{Claudio Cicconetti\,\orcidlink{0000-0003-4503-4223}}\affiliation{Institute for Informatics and Telematics (IIT), National Research Council (CNR), Pisa, Italy}

\date{\today}

\begin{abstract}
Linear quantum repeater chains based on Werner-state purification (the BBPSSW protocol) and entanglement swapping (the BDCZ scheme) are fundamental to entanglement distribution in quantum networks. However, they operate under a stringent operation reliability threshold and rely on resource-intensive recurrence purification rounds, each consuming two entangled pairs to probabilistically produce one. In the literature, hyperentanglement has been proposed to exploit multiple degrees of freedom (DOF), such as polarisation and spatial modes, to encode independent entangled states within a single photon pair. This has led to the definition of DAEPP (DOF-Assisted Entanglement Purification Protocol), which we propose to integrate with the BDCZ scheme, resulting in a chain protocol that we call DAHR (DOF-Assisted Hyperentanglement Repeater). The DAEPP step distils the fidelity of a DOF by consuming other(s). In this work, we propose and analyse a DAHR variant which integrates DAEPP at every segment of an end-to-end path combined with BDCZ. We derive a closed-form end-to-end fidelity recursion that embeds single-segment DAEPP into the BDCZ scheme and give a strict resource lower bound for any BDCZ baseline utilising BBPSSW purification to match DAHR's per-segment effective fidelity. At a representative asymmetric operating point informed by prior experiment, we show numerically that matching DAHR’s single-photon-pair performance requires two to three rounds of BBPSSW purification. Additionally, below a critical operation reliability, no amount of BBPSSW rounds matches DAHR's one DAEPP round performance.
\end{abstract}

\maketitle
\section{Introduction}
\label{sec:intro}

Distributing high-fidelity entanglement across long distances is the foundation for enabling quantum internet~\cite{kimble2008quantum,wehner2018quantum, kumar2025quantum,kumar2026making}. Due to channel attenuation, a direct photon transmission decays exponentially with distance. To this, the established solution is to divide the transmission into segments connected by \emph{quantum repeater}~\cite{sangouard2011quantum,muralidharan2016optimal} nodes that perform entanglement \emph{swapping}~\cite{zukowski1993event,zeilinger1997three}. Additionally, entanglement \emph{purification}~\cite{bennett1996purification,deutsch1996quantum} either on segment-level or chain-level after swapping is performed to lift the fidelity of entanglement distribution. The Briegel-D\"ur-Cirac-Zoller framework (also known as ``BDCZ'')~\cite{briegel1998quantum,dur1999quantum} established the canonical fidelity recursion for chained noisy swapping and the corresponding nested-purification protocol accounting for imperfect operations. In its original form, BDCZ describes pairs in single-DOF Werner states and addresses fidelity degradation via the ``BBPSSW'' (Bennett-Brassard-Popescu-Schumacher-Smolin-Wootters) recurrence purification protocol~\cite{bennett1996purification}, in which two low-fidelity pairs are consumed to produce a single higher-fidelity pair conditioned on a stochastic Bell-measurement outcome.

This, in turn, leads to two structural performance limits. First, imperfect quantum operations steeply degrade what the recurrence can achieve, and below a certain threshold, no number of rounds is sufficient. Second, the photon-pair budget per delivered end-to-end ebit grows as $\sim 2^{k}$ in the number of recurrence rounds $k$ of purification, excluding any further potential failures.

In parallel, a separate line of research in hyperentangled photonics~\cite{kwiat1997hyper,deng2017quantum} has developed an alternate primitive called DOF-assisted entanglement purification protocol (DAEPP)~\cite{sheng2010deterministic,sheng2010one,li2010deterministic,hu2021long}. In this scheme, the entanglement carried in one DOF or mode of a hyperentangled pair (e.g.\ spatial mode) is consumed in a one-shot circuit to purify entanglement carried in another DOF (e.g.\ polarisation). In an experiment, Hu \emph{et al.}~\cite{hu2021long} in 2021 reported entanglement purification using only one pair of hyperentangled state in a fiber over an 11-km multicore-fiber link.

DAEPP has not, to our knowledge, been integrated or studied with chain-level swapping in a closed-form recursion. Alternatively, the hyperentanglement-swapping protocol of Sheng \textit{et al.}~\cite{sheng2010complete} performs Bell-state analysis on both DOFs simultaneously but treats only a single segment and under noiseless-operation assumptions. Additionally, hyperentanglement purification protocols (hyper-EPPs)~\cite{ren2013hyperentanglement,ren2014two,wang2016hyperentanglement} have been studied addressing arbitrary-error refinement on a single segment, but are not formulated as part of a chain recursion. Recent theoretical work on noisy repeater protocol~\cite{victora2023entanglement,benchasattabuse2025integrating,mylavarapu2025teleportation,ghosal2025repeater,miguel2023quantum} considers only single-DOF Bell pairs and orthogonal architectural questions. In particular, the Mylavarapu \emph{et al.} analysis of Werner-state network topologies~\cite{mylavarapu2025teleportation} and the Ramiro \textit{et al.} W-state repeater~\cite{miguel2023quantum} are closely related works but use different state classes and quantifiers.

In this work, we make the following contributions.
\begin{enumerate}
\item An \emph{exact two-fidelity per-segment model} for the DAEPP step that generalises D\"ur et al.'s imperfect-operation BBPSSW recurrence to unequal input fidelities, recovering the two-state bit-flip formula and the symmetric recurrence as limits.
\item A \emph{closed-form end-to-end fidelity recursion} for quantum repeater chains (a DAHR variant) that embeds the DAEPP into the BDCZ scheme.
\item A \emph{strict photon-resource lower bound} on any BDCZ baseline matching DAHR's per-segment fidelity, with native (state-dependent) success probabilities.
\item \emph{Numerical evaluation} at $(\Fpz,\Fsz)=(0.85,0.95)$, $\eBSA=p_1=p_2=0.99$, in an ideal-operation reliability and a matched-operation reliability regime, with a resource-cost comparison and a per-round advantage map to study and compare DAHR quantitatively with established entanglement distribution schemes.
\end{enumerate}

\section{Preliminaries}
\label{sec:prelim}

\subsection{BDCZ chain swap recursion and BBPSSW recurrence purification}\label{sec.bbpssw}

Consider $N$ segments of single-DOF Werner pairs of fidelity $F$. Then entanglement swapping on $N-1$ intermediate nodes with imperfection parameters  $p_1,p_2,\eBSA$ gives the BDCZ end-to-end fidelity as~\cite{briegel1998quantum,dur1999quantum}
\begin{equation}
F_N = \tfrac14\!\left[1 + 3\,(p_1p_2)^{N-1}\!\left(\tfrac{4\eBSA^2-1}{3}\right)^{\!N-1}\!\left(\tfrac{4F-1}{3}\right)^{\!N}\right].
\label{eq:bdcz}
\end{equation}

The general expression for BBPSSW with imperfect operations involving two different Werner fidelity inputs, that is, fidelity $\FS$ for the source/kept pair and fidelity $\FP$ for the partner/measured pair, is
\begin{equation}\label{eq:twofid}
F'(\FS,\FP;\emeas,\pg) = \frac{N(\FS,\FP, \emeas, \pg)}{D(\FS,\FP, \emeas, \pg)},
\end{equation}
with, writing $w_0=\emeas^2+(1-\emeas)^2$, $w_1=2\emeas(1-\emeas)$, $m=(1-\pg^2)/(8\pg^2)$,
\begin{equation}
\label{eq:Nfinal}
\begin{split}
N ={}&
\Bigl[\FS \FP+\tfrac{(1-\FS)(1-\FP)}{9}\Bigr]w_0
\\
&+\Bigl[\tfrac{\FS(1-\FP)}{3}
+\tfrac{(1-\FS)(1-\FP)}{9}\Bigr]w_1 + m,
\end{split}
\end{equation}
\begin{equation}
\label{eq:Dfinal}
\begin{split}
D ={}&
\Bigl[\FS \FP
+\tfrac{\FS(1-\FP)+\FP(1-\FS)}{3}
+\tfrac{5(1-\FS)(1-\FP)}{9}\Bigr]w_0
\\
&+\Bigl[\tfrac{(1-\FS)(1-\FP)}{9}
+\tfrac{\FS(1-\FP)+\FP(1-\FS)}{6}\Bigr]4w_1
+4m.
\end{split}
\end{equation}
where, $\emeas$ is the measurement reliability parameter and $\pg$ is the gate reliability parameter in application of \textsc{cnot} during purification. See Appendix~\ref{app:twofid} for the detailed derivation of the general expression in Eqs.~\eqref{eq:twofid}--\eqref{eq:Dfinal}. The denominator in Eq.~\eqref{eq:Dfinal} is the native acceptance probability up to the $\pg^2$ rescaling, that is, $p_{\rm acc}=\pg^2 D$. So the fidelity and resource counting are computed from the same quantity. Also, the symmetric map, that is, involving equal Werner fidelity inputs, can be written as $T(F)\equiv F'(F,F;\emeas,\pg)$.

\subsection{Hyperentanglement and the DAEPP step}
\label{sec:depp}
In general, hyperentanglement in degrees of freedom up to three is possible~\cite{kwiat1997hyper,zhao2023generation}, which in turn unlocks more complex and enhanced DAEPP step involving more than one purification round with a single hyperentangled pair. While such a DAEPP step would be an interesting study on its own, we choose a simpler DAEPP step by considering hyperentanglement pair in two degrees of freedom (polarisation-spatial). This enables evaluation in an experimentally grounded regime~\cite{hu2021long} and comparison with established single-DOF-based protocols. A polarization--spatial hyperentangled pair is prepared as
\begin{equation}
|\Phi_{\HE}\rangle = |\Phi\rangle_P \otimes |\phi\rangle_S.
\label{eq:HE}
\end{equation}
Similar to the 1-DOF case, one photon to each segment endpoint is propagated. The Li circuit~\cite{li2010deterministic} and the experimental realisation by Hu \emph{et al.}~\cite{hu2021long} apply an intra-photon bilateral \textsc{cnot} between the spatial DOF (control) and polarisation DOF (target) and measure the spatial mode, retaining the polarisation pair on a heralded coincidence. This is the same bilateral-\textsc{cnot} purification step used in BBPSSW (Sect.~\ref{sec.bbpssw}, Appendix~\ref{app:twofid}), with the polarisation pair as source/kept and the spatial pair as partner/measured. Finally, measurement of spatial DOF distils the polarisation DOF.

\begin{figure*}[t]
\centering
\includegraphics[width=0.9\textwidth]{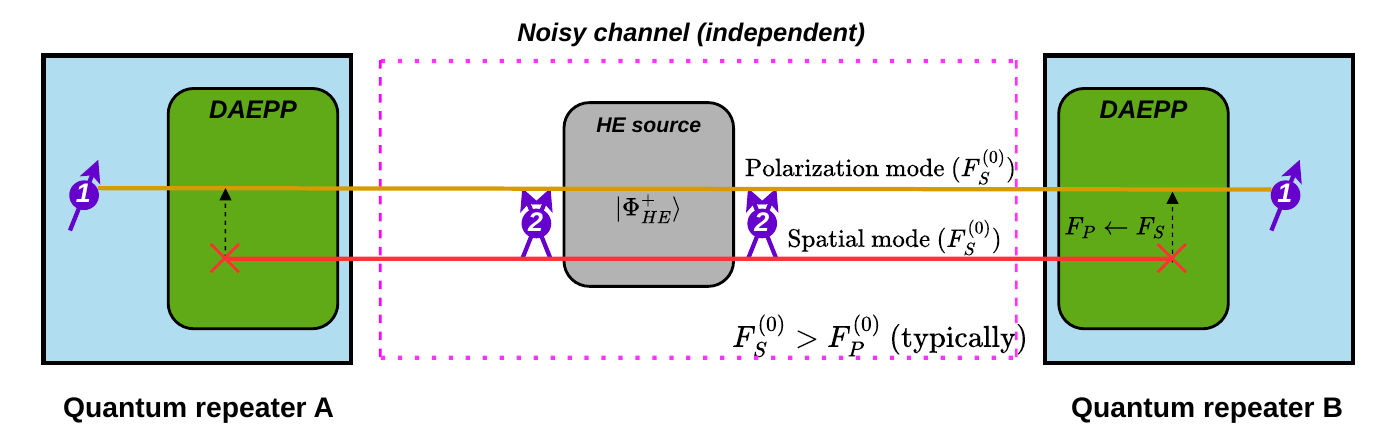}
\caption{The DAHR per-segment workflow. Left to right follows: quantum repeater A, the noisy channel, and quantum repeater B. Hyperentangled pairs are generated by a source in the middle (\textit{MidPointSource} scheme \cite{kumar2025quantum}). The spatial degree of freedom (in red) is consumed to upgrade the polarisation mode before chain-level entanglement swapping.}
\label{fig:schematic}
\end{figure*}

\paragraph*{Regime positioning.}
Sheng \& Deng~\cite{sheng2010deterministic} and Li~\cite{li2010deterministic} idealise the spatial DOF as $\FS=1$ and obtain deterministic purification $\FP'\to1$. Sheng \& Deng~\cite{sheng2010deterministic} utilised two cross-Kerr steps (spatial DOF for the bit-flip sector, frequency DOF for the phase-flip sector) while Li~\cite{li2010deterministic} and Sheng \& Deng~\cite{sheng2010one} utilised a single passive linear-optics step that transfers the full spatial Bell structure onto polarisation. Hu \emph{et al.}~\cite{hu2021long} realised the architecture over 11 km with a single spatial parity check that corrects the bit-flip sector, in the ideal-operation reliability limit ($\emeas=\pg=1$) and utilising a two-state $\{\Phi^+, \Psi^+\}$ input case rather than four. For this, the distilled fidelity with two different fidelity inputs is,
\begin{equation}
\Fstar^{\rm bf} = \frac{\FP\FS}{\FP\FS+(1-\FP)(1-\FS)}.
\label{eq:hu}
\end{equation}
While for a full Werner input, the same ideal-operation reliability limit reduces the generalised expression (Eqs.~\eqref{eq:twofid}--\eqref{eq:Dfinal}) to Eq.~\eqref{eq:fstar-ideal}.  We adopt this experimentally grounded single round as our DAEPP primitive while using in this work the more general and more conservative \emph{Werner} model: after a standard bilateral twirl, each DOF is a Werner state, and the DAEPP output is the full two-fidelity map Eq.\eqref{eq:twofid}. The use of twirl (random correlated bilateral rotations averaging an arbitrary Bell-diagonal state to the isotropic Werner form of equal fidelity) is the standard~\cite{bennett1996mixed} that makes the chain recursion tractable. It is fidelity-preserving and can only discard, never create exploitable correlation. So every fidelity we report is a lower bound on what a model retaining the full Bell-diagonal structure would give. We use it for both DAHR and the BDCZ baseline, so the comparison is unaffected by the choice.

\begin{lemma}[Exact two-fidelity DAEPP output]
\label{lem:depp}
Let the per-segment hyperentangled pair, after channel noise and a bilateral twirl, be a product of Werner states with polarisation fidelity $\FP$ (kept) and spatial fidelity $\FS$ (measured). Then one DEAPP round after a standard bilateral twirl outputs a Werner polarisation pair of fidelity
\begin{equation}
\Fstar = F'(\FP,\FS;\emeas,\pg),
\label{eq:fstar}
\end{equation}
given by Eqs.\eqref{eq:twofid}--\eqref{eq:Dfinal}, retained with heralded probability $p_{\rm DAEPP}=\pg^2 D(\FP,\FS)$. In the ideal-operation reliability limit $\emeas=\pg=1$,
\begin{equation}
\Fstar = \frac{\FP\FS+\tfrac{(1-\FP)(1-\FS)}{9}}{\FP\FS+\tfrac{\FP(1-\FS)+\FS(1-\FP)}{3}+\tfrac{5(1-\FP)(1-\FS)}{9}} .
\label{eq:fstar-ideal}
\end{equation}
\end{lemma}
\begin{proof} Refer Appendix~\ref{app:twofid}.
\end{proof}

\begin{remark}[Limits]
\label{rem:limits}
(a) $\FP=\FS=F$ in Eq.~\eqref{eq:fstar} gives D\"ur's symmetric BBPSSW recurrence~\cite{dur1999quantum}. (b) In the bit-flip-only restriction, Eq.~\eqref{eq:fstar} collapses to the Hu formula Eq.~\eqref{eq:hu} and, as $\FS\to1$, $\Fstar\to1$.
(c) Under the full Werner model, $\FS\to1$ gives only $\Fstar\to 3\FP/(2\FP+1)<1$. A single round here clears the bit-flip (amplitude) sector but leaves the kept pair's phase errors. Note that the same holds for the operation of one BBPSSW round. Convergence to unity fidelity utilises multiple rounds, which comes from iterating with re-twirling.
\end{remark}

\section{System model and protocol}
\label{sec:protocol}

\subsection{Channel, noise, and parameters}
The model in consideration is a linear chain having $N$ segments, each emitting one hyperentangled pair (Eq.~\eqref{eq:HE}). Independent channels act on the two DOFs; that is, after a bilateral twirl, each DOF is a Werner state, with post-channel polarisation and spatial fidelities $\Fpz,\Fsz$. Typically $\Fsz>\Fpz$, the path/spatial DOF is less exposed to the birefringence and mode-dependent dephasing that degrade polarisation in fiber~\cite{sheng2010deterministic,sheng2010one,li2010deterministic,hu2021long}, making it the more robust mode and hence the natural high-fidelity reference to measure against. Purification operations (DAEPP and BBPSSW alike) carry measurement reliability $\emeas$ and gate reliability $\pg$, while the swap nodes carry $\eBSA,p_1,p_2$ as in Eq.~\eqref{eq:bdcz}.

\subsection{Per-segment workflow}\label{ssec:persegment}
For each segment (see Fig.~\eqref{fig:schematic}), the steps are as follows:
\begin{enumerate}
    \item emit $|\Phi^+_{\HE}\rangle$ from the source,
    \item transmit via quantum channel, giving twirled Werner states $(\Fpz,\Fsz)$,
    \item run one DAEPP round at one endpoint, keeping polarisation and measuring the spatial mode, producing a Werner polarisation pair of fidelity $\Fstar=F'(\Fpz,\Fsz;\emeas,\pg)$ (Lemma~\ref{lem:depp}), retained with probability $p_{\rm DAEPP}$.
\end{enumerate}
Then, the $N$ purified segments undergo standard polarisation-only chained swapping with feed-forward Pauli corrections, exactly as in the BDCZ scheme.

The architecture is intentionally lean: one hyperentangled pair per segment, one heralded DAEPP round consuming the spatial DOF (not a second pair), and standard chain swapping thereafter.

\section{Recursion and Resource Analysis}
\label{sec:analysis}

\subsection{End-to-end fidelity}
\begin{theorem}[DAHR end-to-end fidelity]
\label{thm:dahr-recursion}
Under the model of Sec.~\ref{sec:protocol}, for $N\geq1$, similar to model in Eq.~\eqref{eq:bdcz},
\begin{align}
F_N^{\DAHR} = \tfrac14\!\left[1 + 3(p_1p_2)^{N-1}\!\left(\tfrac{4\eBSA^2-1}{3}\right)^{\!N-1}\!\left(\tfrac{4\Fstar-1}{3}\right)^{\!N}\right]\!,
\label{eq:dahr-recursion}
\end{align}
with $\Fstar=F'(\Fpz,\Fsz;\emeas,\pg)$ from Lemma~\ref{lem:depp}.
\end{theorem}
\begin{proof}
The end-to-end fidelity is essentially a composition of the per-segment DAEPP round with the BDCZ chain. After step 3 of Sec.~\ref{ssec:persegment}, segment $i$ essentially holds a polarization pair in the Werner state $\rho(\Fstar)$ with $\Fstar=F'(\Fpz,\Fsz;\emeas,\pg)$, by Lemma~\ref{lem:depp}. The $N$ post-DAEPP pairs are statistically independent; in other words, the channel noise on segment $i$ acts on photonic Hilbert spaces disjoint from those of segment $j\neq i$. And at each segment, the DAEPP acts locally at one endpoint (see Fig.~\eqref{fig:schematic}), commuting with the channel acting on the partner photon, so no correlations are induced between segments. The chain entering the swap stage is therefore the product state $\bigotimes_{i=1}^{N}\rho(\Fstar)$. Therefore, the swap stage is identical to the BDCZ chain on single-DOF Werner pairs of starting fidelity $\Fstar$, with swap-node parameters $\eBSA,p_1,p_2$. So Eq.~\eqref{eq:bdcz} applies the map $F\mapsto\Fstar$, giving Eq.~\eqref{eq:dahr-recursion}. As discussed in Sec.~\ref{sec:depp}, the twirl to Werner form before swapping is the standard BDCZ convention and is fidelity-preserving (it fixes the singlet fraction $\Fstar$); it discards only correlation information not used by the chain recursion, so Eq.~\eqref{eq:dahr-recursion} is a conservative estimate.
\end{proof}
\begin{remark}[Versus raw BDCZ]
The raw-BDCZ recursion follows from $\Fstar\mapsto\Fpz$ in Theorem~\eqref{thm:dahr-recursion}. Since $\Fstar>\Fpz$ at the operating point, the per-segment factor $(4\Fstar-1)/3$ exceeds $(4\Fpz-1)/3$ and the gap is amplified as the $N$-th power.
\end{remark}

\subsection{Role asymmetry and the spatial-robustness argument}
\label{sec:asymmetry}
The denominator (Eq.~\eqref{eq:Dfinal}) is symmetric under $\FS\leftrightarrow \FP$ (an acceptance probability cannot depend on which pair is privately kept), but the numerator (Eq.~\eqref{eq:Nfinal}) is not. The $w_1$ term in the numerator ($\propto \FS(1-\FP)$) is an accept case in which the kept pair is $\ket{\Phip}$ and the discarded pair is $\ket{\Psip}$, resulting from measurement error in exactly one of the measured qubits (Appendix~\ref{app:twofid}). The exact gap due to this is given by:
\begin{equation}
\Delta F^{'}=F'(\FS,\FP)-F'(\FP,\FS) = \frac{2\emeas(1-\emeas)}{3\,D}\,(\FS-\FP).
\label{eq:asym}
\end{equation}
From this, two architectural design insights can be drawn. First, since $D>0$, the gap has the sign of $\FS-\FP$: \emph{it pays to keep the higher-fidelity pair ($\FS$)}. Second, the magnitude is controlled by $\emeas(1-\emeas)$, vanishing as $\emeas\to1$, that is, with a perfectly reliable measurement, and the slots become interchangeable.

However, in this work, DAHR's role assignment is a structural constraint, not a fidelity optimisation. The polarisation mode carries the repeater qubit through the swap chain and must be kept, while the spatial mode is the auxiliary degree of freedom within the same hyperentangled pair and is therefore consumed by the parity measurement~\cite{li2010deterministic,sheng2010deterministic,sheng2010one,hu2021long}. Within this fixed assignment, two distinct features of the generalised two-fidelity relation enter, with very different magnitudes.

The spatial mode is the more robust DOF in fiber, $\FS>\FP$, so DAHR measures against a high-fidelity reference, whereas BBPSSW, having no second DOF can only measure another polarisation copy at fidelity $\FP$. At the operating point $(\FP, \FS ) = (0.85, 0.95)$ with matched operation reliability $\emeas = \pg = 0.99$, the same kept polarisation is purified to a strictly higher single-round output,
\begin{equation}
\underbrace{F'(\FP,\FS)}_{\text{DAHR}}=0.915 \;>\; \underbrace{F'(\FP,\FP)}_{\text{BBPSSW}}=0.874,
\label{eq:perround}
\end{equation}
which is a $\sim+0.041$ per-round advantage. DAEPP consumes one DOF while BBPSSW a whole extra pair. This is the entire source of DAHR's per-round gain over BBPSSW, and it is the substantive content of the spatial-robustness argument.

A separate, sub-leading effect is the role asymmetry derived in Eq.~\eqref{eq:asym}. Read as a free optimisation, it would in fact prescribe keeping the higher-fidelity pair and measuring the lower one, which, with $\FS>\FP$, would mean keeping the spatial pair and measuring polarisation. DAHR cannot adopt that configuration without breaking the protocol, since polarisation is the data carrier. The cost of being locked into the keep-noisier-pair slot is therefore the role-swap gap measured by Eq.~\eqref{eq:asym} which at $\emeas=\pg=0.99$ and the operating point evaluates to $\sim 7\times10^{-4}$, two orders of magnitude smaller than the per-round gain Eq.~\eqref{eq:perround}, and vanishes entirely as readout improves. The literature's prescription \emph{keep polarisation, measure spatial} is thus the protocol-determined assignment \cite{sheng2010deterministic,sheng2010one,li2010deterministic,hu2021long}, and the model quantifies it as essentially optimal. The asymmetry penalty incurred for being forced into the slot the free-optimisation argument would have rejected is negligible against the dominant measured-fidelity gain.

\subsection{Photon-resource lower bound vs.\ BDCZ-with-BBPSSW purification}\label{ssec:resource}
\begin{figure*}[t]
\centering
\includegraphics[width=0.9\textwidth]{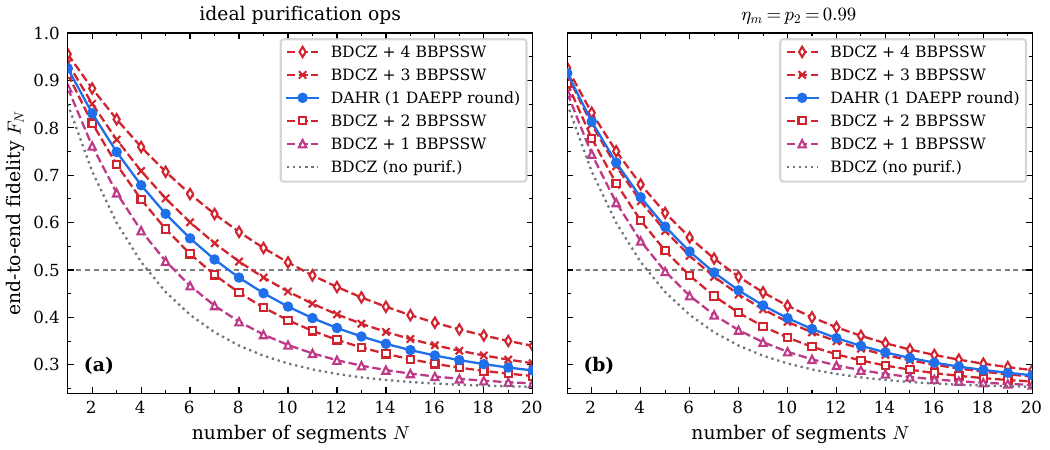}
\caption{End-to-end fidelity $F_N$ vs chain length at the operating point $(\Fpz,\Fsz)=(0.85,0.95)$, with BDCZ parameters as $\eBSA=p_1=p_2=0.99$. (a) ideal purification operation reliability ($\emeas=\pg=1$); (b) matched operation reliability ($\emeas=\pg=0.99$). DAHR (one DAEPP round) stays above BDCZ-with-2-BBPSSW for the ideal purification operation reliability case and is even marginally better than BDCZ-with-3-BBPSSW for the matched operation reliability case across the number of segment $N$ range.}
\label{fig:fidelity-vs-N}
\end{figure*}

\begin{figure}[t]
\centering
\includegraphics[width=\columnwidth]{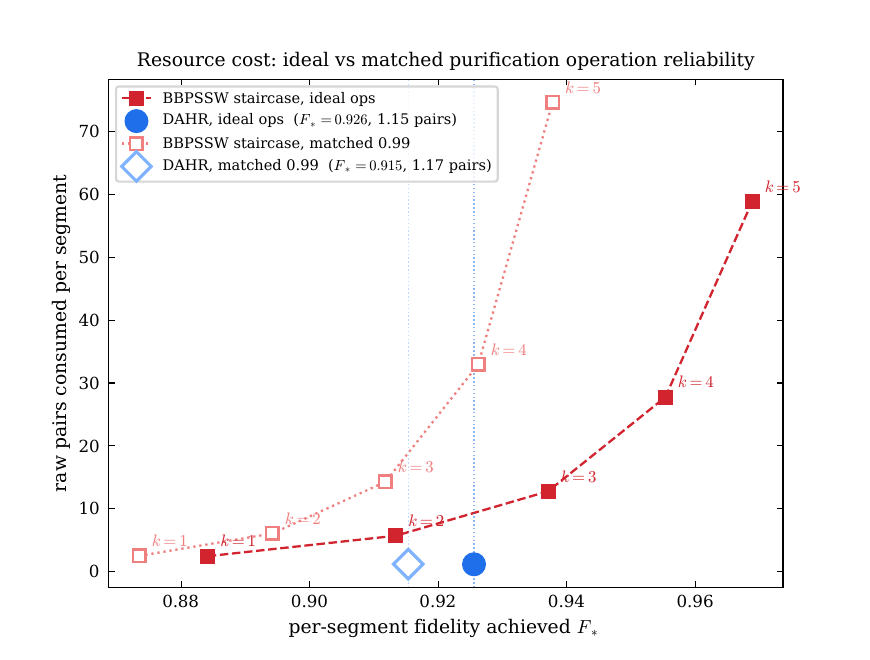}
\caption{Raw pairs consumed per segment versus delivered per-segment fidelity for the ideal and matched purification-operation reliability regimes. In the ideal regime, DAHR (one DAEPP, blue, solid circle) achieves $(0.926, 1.15)$, while the BBPSSW staircase (red, solid square, $k=1\ldots5$) requires $k=2$–$3$ rounds to reach comparable performance, corresponding to $5.7$–$12.8$ raw pairs. In the matched-operation reliability regime, DAHR (one DAEPP, blue, hollow diamond) attains $(0.915, 1.17)$, whereas BBPSSW (red, hollow square) must exceed $k=3$ rounds to approximately match this point, requiring about $15$ raw pairs.
}
\label{fig:budget}
\end{figure}

To match DAHR's per-segment fidelity with a BDCZ baseline one applies symmetric BBPSSW rounds, each an instance of the map (Eqs.~\eqref{eq:twofid}--\eqref{eq:Dfinal}) at $\FS=\FP$.
\begin{theorem}[Photon-resource lower bound]
\label{thm:resource}
Let $T(F)=F'(F,F;\emeas,\pg)$ and $p_{\rm acc}(F)=\pg^2 D(F,F)$, and let $k^*$ be the least integer with $T^{k^*}(\Fpz)\geq\Fstar$ ($k^*=\infty$ if none). Any BDCZ baseline built from symmetric BBPSSW recurrence rounds that matches DAHR's per-segment $\Fstar$ consumes
\begin{equation}
n_{\BDCZ}^{\rm per\,seg} \geq \prod_{r=0}^{k^*-1}\frac{2}{p_{\rm acc}(T^r(\Fpz))} \geq 2^{k^*}
\label{eq:resource-bound}
\end{equation}
polarisation pairs per segment. DAHR uses one hyperentangled pair, retained with $p_{\rm DAEPP} = \pg^2 D(\FP, \FS)$, i.e.\ $1/p_{\rm DAEPP}$ pairs per delivered segment.
\end{theorem}
\begin{proof}
Each round consumes two pairs of fidelity $T^{r}(\Fpz)$ and yields one of fidelity $T^{r+1}(\Fpz)$ with native probability $p_{\rm acc}\leq1$. Hence, the expected photon count to reach level $k^*$ is the stated product. However, if the symmetric map's upper fixed point ($F^\dagger$)\cite{dur1999quantum} lies below what DAHR achieves, that is, $F^\dagger<\Fstar$, then $k^*=\infty$.
\end{proof}

\begin{corollary}[Unbounded-cost regime]
\label{cor:unbounded}
If $\emeas,\pg<1$ and $\Fstar>F^\dagger$, no finite BDCZ schedule matches DAHR's per-segment fidelity while DAHR remains at $1/p_{\rm DAEPP}$ pairs per segment as represented in Figure~\ref{fig:corollary2} and further discussed in Sec.~\ref{sec:numerics}. 
\end{corollary}

\begin{figure*}[t]
\centering
\includegraphics[width=0.9\textwidth]{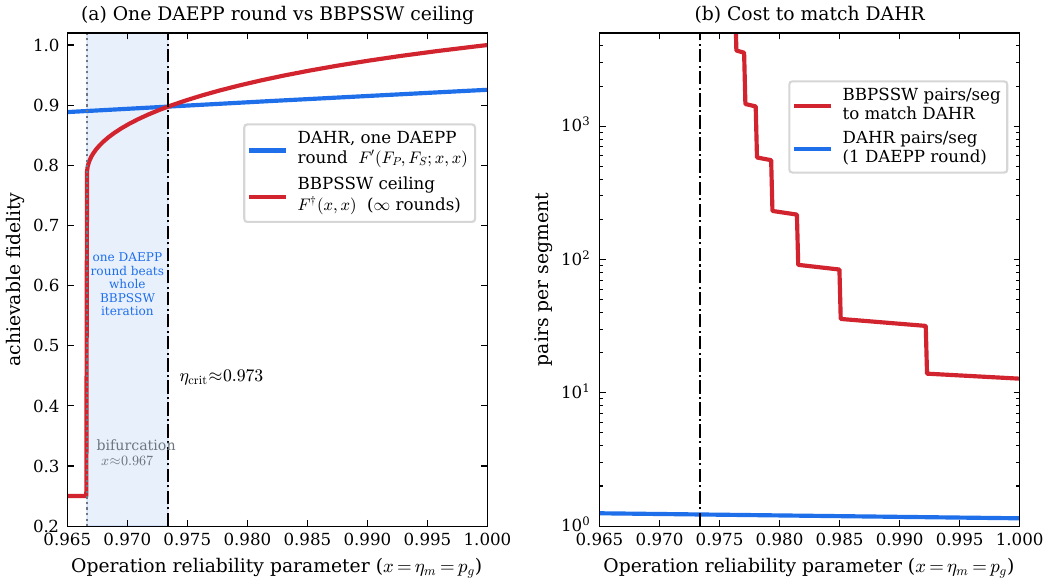}
\caption{For Corollary~\ref{cor:unbounded} with DAHR and BBPSSW at the same operation reliability $x=\emeas=\pg$. (a) Fidelity head-to-head for DAHR's one DAEPP round $F'(\FP,\FS;x,x)$ (blue) and BBPSSW's $\infty$-round ceiling $F^\dagger(x,x)$ (red). The two cross at $\eta_{\rm crit}\approx0.973$, and the shaded band $0.967<x<\eta_{\rm crit}$ (above the bifurcation $\approx0.967$) is the window where one DAEPP round beats BBPSSW's whole iteration. (b) Above $\eta_{\rm crit}$ BBPSSW can match DAHR but at exploding photon cost. Pairs/segment to match diverges as $x\to\eta_{\rm crit}^+$, against DAHR's flat $\sim1.15$ pairs/segment.}
\label{fig:corollary2}
\end{figure*}

\section{Numerical Evaluation}
\label{sec:numerics}
We select a representative operation point for evaluation at $\Fpz=0.85$, $\Fsz=0.95$, $\eBSA=p_1=p_2=0.99$, in two regimes: \emph{ideal operation reliability} ($\emeas=\pg=1$, the protocol ceiling) and \emph{matched operation reliability} ($\emeas=\pg=0.99$). Note that it is common to assume equal measurement and gate reliability parameters and to use a single parameter to characterise the overall operation reliability~\cite{dur1999quantum}. The operating point for evaluation is selected based on the experimental working values from~\cite{hu2021long}, which reports the fidelities of hyperentanglement state distributed over 11 km in the multi-core fiber (MCF) of around $\sim 0.95$ for both polarisation and spatial modes before the introduction of bit-flip errors. We keep the same value for the spatial mode and a lower value for the polarisation mode, considering the spatial-robustness argument as discussed in Sec.~\ref{sec:asymmetry}. All numerics come from the implementation of Eqs.~\eqref{eq:twofid}-\eqref{eq:Dfinal} and Eq.~\eqref{eq:bdcz} with native success probabilities, validated against the derivation checks (Appendix~\ref{app:twofid}).

\paragraph*{Per-segment DAEPP.} One DAEPP round gives $\Fstar=0.926$ (ideal operation reliability) with $p_{\rm DAEPP}=0.873$, i.e.\ $1.15$ HE pairs per delivered segment and $\Fstar=0.915$ (matched operation reliability) with $p_{\rm DAEPP}=0.851$, i.e.\ $1.17$ HE pairs per delivered segment. While a single BBPSSW round on two $\FP=0.85$ copies gives only $\Fstar = 0.884$ (ideal operation reliability) and $\Fstar = 0.874$ (matched operation reliability).

\paragraph*{End-to-end fidelity.} Figure~\ref{fig:fidelity-vs-N} shows $F_N$ for raw BDCZ, BDCZ with $k=1,2,3,4$ BBPSSW rounds, and DAHR for both ideal purification operation reliability and matched operation reliability cases. At $N=8$, $F_8^{\DAHR}=0.484$ (ideal) and $0.457$ (matched), above BDCZ with two BBPSSW rounds throughout for the ideal purification operation reliability case and even marginally better than BDCZ with three BBPSSW rounds for the matched operation reliability case.

\paragraph*{Resource cost.} Figure~\ref{fig:budget} plots raw pairs per segment consumed against the per-segment fidelity delivered under ideal and matched purification-operation reliability regimes. In the ideal case, DAHR achieves $\Fstar = 0.926$ at $1.15$ raw pairs per segment. The BBPSSW staircase reaches $0.913$ at $k=2$ corresponding to $5.7$ pairs, and $0.937$ at $k=3$ corresponding to $12.8$ pairs. In the matched-operation reliability regime, DAHR attains $\Fstar = 0.915$ at $1.17$ pairs per segment, while BBPSSW must exceed $k=3$ to reach comparable fidelity, requiring approximately $15$ pairs. Across both regimes, matching DAHR requires the BDCZ baseline to expend roughly $5$ to $11$ times more photons per segment.

\paragraph*{BBPSSW ceiling (Corollary~\ref{cor:unbounded}).}
Figure~\ref{fig:corollary2} plots the head-to-head between DAHR and BBPSSW on a single imperfection parameter axis obtained by setting $x=\emeas=\pg$ (operation reliability parameter). Panel (a) shows DAHR's one-round DAEPP output $F'(\FP,\FS;x,x)$ with respect to the BBPSSW infinite-round ceiling $F^\dagger(x,x)$. The two cross at $\eta_{\rm crit}\approx0.973$. For $0.967<x<\eta_{\rm crit}$ (the band shaded above the BBPSSW bifurcation at $\approx0.967$) a single DAEPP round beats the entire BBPSSW iteration, while above $\eta_{\rm crit}$ BBPSSW's ceiling exceeds DAHR's one-round output. While Panel (b) shows what that ``BBPSSW wins above $\eta_{\rm crit}$'' actually costs. The photons/segment BBPSSW must spend to \emph{match} DAHR's one-round output diverges as $x\to\eta_{\rm crit}^+$, climbing from $\sim13$ at $x=1$ through $33$ at $x=0.99$, $227$ at $x=0.98$, to $\infty$ at $\eta_{\rm crit}$, while DAHR remains at $\sim1.15$ throughout.

\begin{figure}[t]
\centering
\includegraphics[width=\columnwidth]{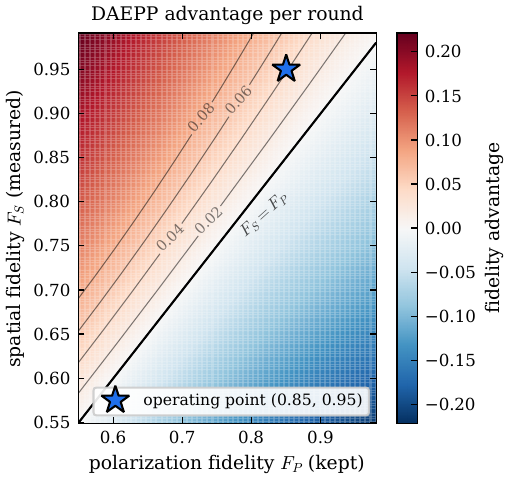}
\caption{Single-round DAEPP advantage $F'(\FP,\FS)-F'(\FP,\FP)$ over the spatial-polarization fidelity plane (ideal operation reliability). The region is Positive (red) wherever $\FS>\FP$. The star here marks the represented operating point in the work, that is, $(0.85,0.95)$. The diagonal $\FS=\FP$ is the BBPSSW-equivalent line.}
\label{fig:crossover}
\end{figure}
\paragraph*{Per-round advantage map.}
Figure~\ref{fig:crossover} maps the single-round advantage $F'(\FP,\FS)-F'(\FP,\FP)$ of DAEPP over BBPSSW across the $(\FP,\FS)$ plane. It is positive precisely where the spatial mode is the more robust DOF, $\FS>\FP$ (above the diagonal), and the operating point sits at $+0.041$. The advantage is the direct payoff of having the polarisation mode against a robust reference rather than a noisy copy of itself.

\section{Summary and Outlook}
\label{sec:conclusion}
In this work, we proposed and analysed a DAHR (DOF-Assisted Hyperentanglement Repeater chain protocol) variant which utilises an additional degree of freedom for the entanglement purification step required in long-distance entanglement distribution. Each segment in a linear quantum repeater chain network operating under DAHR undergoes DAEPP (DOF-Assisted Entanglement Purification Protocol) where intra-photon operation accompanied by measurement on the robust spatial degree of freedom purifies the distributed entanglement in polarisation mode. 

We generalised BBPSSW with the imperfect operations recurrence expression of D\"ur et al. to two unequal Werner fidelities corresponding to two different degrees of freedom of a hyperentanglement pair and embedded it in closed form into the BDCZ chain recursion (Theorem~\ref{thm:dahr-recursion}). Because the same map reduces to D\"ur et al.'s expression of the BBPSSW recurrence at equal fidelities, DAHR and the BDCZ baseline are compared inside one first-principles model. The derived generalised map is asymmetric under $\FS\leftrightarrow \FP$ due to asymmetry in the numerator, which comes from the $ \wone$-weighted term that corresponds to the case of parity match in the purification procedure where exactly one of the qubits of the partner pair undergoes flipping. Following the spatial-robustness argument, comparison between DAHR with a single DAEPP round and BDCZ with multiple BBPSSW purification rounds at a representative asymmetric operating point, with fidelity values at the level demonstrated by the 11-km multicore-fiber distribution experiment of a hyperentangled state. In the analysis, we discover that DAHR, operating with just a single DAEPP round, outperforms BDCZ, even with multiple purification rounds across the number of segments involved in the linear chain of repeaters, while consuming one DOF rather than a whole pair. To match or beat DAHR, the baseline BDCZ with BBPSSW need $5$--$11$ times more photons per segment.

The analysis presented in this work connects work in hyperentanglement with quantum networking, which at the moment only utilises single DOF freedom for entanglement distribution in quantum networks. This opens the door to exploration of the usefulness of more than one DOF in quantum networking protocols and architectures, including engineering challenges in routing, optimisation, and resource allocation. Other open directions include DAHR variant which utilises conjugate-basis two-pass DAEPP for general Werner channels, which would clear the residual phase error on top of bit-flip error. The open question in that case would be resource counting, since the second pass would either require a second robust auxiliary DOF (e.g., frequency or time-bin as in \cite{sheng2010deterministic}) or a second hyperentanglement pair altogether. Further exploration could be of an HBSA (Hyperentanglement Bell state analysis \cite{sheng2010complete}) variant instead of BSA, deferring DAEPP to once per chain, and multicore-fiber tests at multi-segment networks ($N=2$--$4$).

\appendix

\section{Derivation of the two-fidelity recurrence}
\label{app:twofid}
In this section, we derive the general map of Eqs.~\eqref{eq:twofid}--\eqref{eq:Dfinal} by generalizing
D\"ur's symmetric BBPSSW recurrence~\cite{dur1999quantum} to two \emph{distinct}
Werner fidelities, that is $\FS\neq \FP$ (Figure~\ref{fig:distinct_fidelity}), under imperfect operations. 

Each input pair can be written as a
Bell-diagonal state in the convention
$A=P(\Phip)$, $B=P(\Psip)$, $C=P(\Psim)$, $D=P(\Phim)$, such that,
\begin{equation}
\rho_i = A_i\,\ketbra{\Phip}{\Phip}
       + B_i\,\ketbra{\Psip}{\Psip}
       + C_i\,\ketbra{\Psim}{\Psim}
       + D_i\,\ketbra{\Phim}{\Phim},
\label{eq:rhoi}
\end{equation}
for $i\in\{s,p\}$, the source/kept pair (qubits $1,2$) and the partner/measured
pair (qubits $3,4$) as depicted in Figure~\ref{fig:distinct_fidelity}. For Werner inputs of fidelity $F_i$,
\begin{equation}
A_i=F_i,\qquad B_i=C_i=D_i=\tfrac{1-F_i}{3},
\label{eq:wernersub}
\end{equation}
so that $\rho_s=\FS\ketbra{\Phip}{\Phip}+\tfrac{1-\FS}{3}(\ketbra{\Psip}{\Psip}
+\ketbra{\Psim}{\Psim}+\ketbra{\Phim}{\Phim})$ and similarly for $\rho_p$ with
$\FP$. We keep track of the distinct Bell pair combinations as in Table~\ref{tab:master}.
\begin{figure}[t]
\centering
\includegraphics[width=\columnwidth]{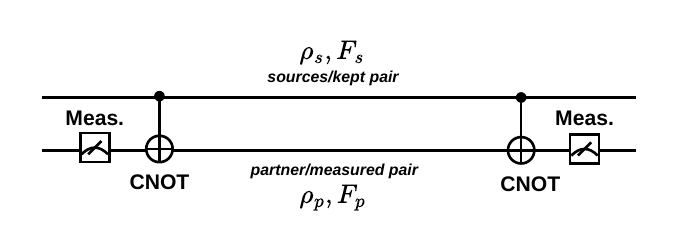}
\caption{Circuit diagram for the entanglement purification consuming two distinct Werner fidelities $\rho_s$ and $\rho_p$. The purification procedure consists of one two-qubit operation (\textsc{cnot}) and one single-qubit operation (Meas.).}
\label{fig:distinct_fidelity}
\end{figure}
While working with imperfections in the measurement operation, involving the two $\sigma_z$ readouts on the partner qubits, each flip with probability
$1-\emeas$. That is, for a single measurement apparatus employed, it would give a correct measurement result with $\eta_m$ probability and an incorrect result with $1-\eta_m$ probability. Considering measurement on both qubits of the partner pair with respect to the source pair, which undergoes no measurement operation, the overall outcome can be categorised in the following two ways. 
\begin{enumerate}\label{ref.points}
    \item If the measurement apparatus employed for qubit measurement allocates the correct basis of measurement, then the measurement outcome would be correct. This happens when either both of the qubits of the partner pair happen to not flip ($\eta_m^2$ probability) or both of the qubits of the partner pair happen to flip ($(1-\eta_m)^2$ probability). Together, this constitutes the total probability of correct basis selection for the measurement $\wzero = \eta_m^2+(1-\eta_m)^2$.
    \item If exactly one of the qubits of the partner pair undergoes flipping, then there would be a difference in the measurement basis among qubit of entangled pair and hence the wrong outcome is generated. The total probability for the wrong basis selection would be $\wone = 2 \eta_m (1-\eta_m)$.
\end{enumerate}

Table~\ref{tab:master} enumerates the $16$ input combinations, with the Werner
probability column Eq.~\eqref{eq:wernersub} written out, and the keep weight ($\wzero$
for matches, $\wone$ for mismatches) attached to each row. Exactly four rows reach $\Phip$: two genuine matches (rows~1,~16,
weight $\wzero$) and two false accepts (rows~2,~15, weight $\wone$).
\begin{table*}[t]
\caption{Enumeration for two unequal Werner pairs under imperfect
readout. Pair~1 is the source (control, kept); Pair~2 is the partner (target,
measured). The probability column is the Werner value Eq.~\eqref{eq:wernersub} of
that row; the weight is $\wzero=\emeas^2+(1-\emeas)^2$ for parity-match rows
(kept on a correct readout) and $\wone=2\emeas(1-\emeas)$ for
parity-mismatch rows ($a_s=a_p$, kept only on a single readout flip).}
\label{tab:master}
\centering
\small
\renewcommand{\arraystretch}{1.18}
\begin{tabular}{cccccccc}
\toprule
\# & Pair 1 (src) & Pair 2 (ptn) & $P_{\rm row}$ (Werner) & $a_s{=}a_p$ & weight & out & $\to\Phip$\\
\midrule
1  & $\Phip$ & $\Phip$ & $\FS\FP$                       & yes & $\wzero$ & $\Phip$ & $\checkmark$\\
2  & $\Phip$ & $\Psip$ & $\tfrac{\FS(1-\FP)}{3}$         & no  & $\wone$  & $\Phip$ & $\checkmark$\\
3  & $\Phip$ & $\Psim$ & $\tfrac{\FS(1-\FP)}{3}$         & no  & $\wone$  & $\Phim$ & \\
4  & $\Phip$ & $\Phim$ & $\tfrac{\FS(1-\FP)}{3}$         & yes & $\wzero$ & $\Phim$ & \\
5  & $\Psip$ & $\Phip$ & $\tfrac{(1-\FS)\FP}{3}$         & no  & $\wone$  & $\Psip$ & \\
6  & $\Psip$ & $\Psip$ & $\tfrac{(1-\FS)(1-\FP)}{9}$     & yes & $\wzero$ & $\Psip$ & \\
7  & $\Psip$ & $\Psim$ & $\tfrac{(1-\FS)(1-\FP)}{9}$     & yes & $\wzero$ & $\Psim$ & \\
8  & $\Psip$ & $\Phim$ & $\tfrac{(1-\FS)(1-\FP)}{9}$     & no  & $\wone$  & $\Psim$ & \\
9  & $\Psim$ & $\Phip$ & $\tfrac{(1-\FS)\FP}{3}$         & no  & $\wone$  & $\Psim$ & \\
10 & $\Psim$ & $\Psip$ & $\tfrac{(1-\FS)(1-\FP)}{9}$     & yes & $\wzero$ & $\Psim$ & \\
11 & $\Psim$ & $\Psim$ & $\tfrac{(1-\FS)(1-\FP)}{9}$     & yes & $\wzero$ & $\Psip$ & \\
12 & $\Psim$ & $\Phim$ & $\tfrac{(1-\FS)(1-\FP)}{9}$     & no  & $\wone$  & $\Psip$ & \\
13 & $\Phim$ & $\Phip$ & $\tfrac{(1-\FS)\FP}{3}$         & yes & $\wzero$ & $\Phim$ & \\
14 & $\Phim$ & $\Psip$ & $\tfrac{(1-\FS)(1-\FP)}{9}$     & no  & $\wone$  & $\Phim$ & \\
15 & $\Phim$ & $\Psim$ & $\tfrac{(1-\FS)(1-\FP)}{9}$     & no  & $\wone$  & $\Phip$ & $\checkmark$\\
16 & $\Phim$ & $\Phim$ & $\tfrac{(1-\FS)(1-\FP)}{9}$     & yes & $\wzero$ & $\Phip$ & $\checkmark$\\
\bottomrule
\end{tabular}
\end{table*}

The table with respect to the generalised map (Eqs.~\eqref{eq:twofid}--\eqref{eq:Dfinal}) is read as follows:
\paragraph{Numerator (rows $\to\Phip$).}
To reach $\Phip$, in the $\wzero$ branch the partner pair must also be
$\Phi$-type with matching phase (rows~1,~16). While in the $\wone$
branch the partner pair must be $\Psi$-type with the matching phase (rows~2,~15).
Collecting the Werner probability of these four rows we have,
\begin{equation}
\label{eq:N0app}
\begin{split}
N_0 ={}&
\Bigl[\FS \FP+\tfrac{(1-\FS)(1-\FP)}{9}\Bigr]w_0
\\
&+\Bigl[\tfrac{\FS(1-\FP)}{3}
+\tfrac{(1-\FS)(1-\FP)}{9}\Bigr]w_1.
\end{split}
\end{equation}
\paragraph{Denominator (all kept rows).}
Every match row contributes with $\wzero$ and every mismatch row with $\wone$.
Grouping through the parity marginals, we get,
\begin{equation}
\label{eq:D0app}
\begin{split}
D_0 ={}&
\Bigl[\FS \FP
+\tfrac{\FS(1-\FP)+\FP(1-\FS)}{3}
+\tfrac{5(1-\FS)(1-\FP)}{9}\Bigr]w_0
\\
&+\Bigl[\tfrac{(1-\FS)(1-\FP)}{9}
+\tfrac{\FS(1-\FP)+\FP(1-\FS)}{6}\Bigr]4w_1.
\end{split}
\end{equation}

The contrast between Eq.~\eqref{eq:N0app} and Eq.~\eqref{eq:D0app} is the origin of the
role asymmetry. The numerator's $\wone$ term keeps only the two
source-$\Phi$/partner-$\Psi$ rows that reach $\Phip$ (not symmetric under
$s\!\leftrightarrow\!p$), whereas the denominator's $\wone$ term sums all eight
mismatch rows and is symmetric. Ultimately expanding Eq.~\eqref{eq:wernersub} into
Eq.\eqref{eq:N0app}--Eq.\eqref{eq:D0app} and adding the floor (discussed below) gives exactly
Eqs.~\eqref{eq:Nfinal}--\eqref{eq:Dfinal} of the main text.

As shown in Fig.~\ref{fig:distinct_fidelity}, the purification step applies a
bilateral \textsc{cnot} $O_{13}O_{24}$ before the readout, one \textsc{cnot} on
either side. In the
D\"ur \emph{et al.} model~\cite{dur1999quantum} each factor is an ideal
\textsc{cnot} followed by a two-qubit depolariser of reliability $\pg$,

\begin{equation}
O_{ij}=\mathcal{D}_{\pg}\circ\textsc{cnot}_{ij},
\qquad
\mathcal{D}_{\pg}(\rho)=\pg\,\rho+(1-\pg)\,\frac{I_4}{4}.
\label{eq:gate-model}
\end{equation}
So each \textsc{cnot} is applied cleanly with probability $\pg$ and replaces its
pair with white noise with probability $1-\pg^2$. In other words, the two factors are independent, and the outcome
collapses into two branches, that is, with probability $\pg^2$ the ideal enumeration of Table~\ref{tab:master} applies as it is and with
probability $1-\pg^2$ at least one \textsc{cnot} depolarisation factor emerges. In the event of a single depolariser, the pair is maximally mixed, which sends the
whole four-qubit register to $I_{16}/16$. Therefore, the ``one fails'' and ``both
fail'' cases land in the same state and are tracked together.

Now the outcome of the application of imperfect \textsc{cnot} interacts with the imperfection in measurement which was treated above. Since we are interested in the case of accounting for imperfections, we focus on the case where depolarisation occurs with probability $1-\pg^2$. As the failed branch is input-independent, the imperfection
enters in closed form. On depolarisation, the four-qubit register is written as $I_{16}/16=\tfrac{I_4}{4}\otimes\tfrac{I_4}{4}$
(source $\otimes$ partner). Now the acceptance depends only on the partner pair. A
maximally mixed partner pair $\tfrac{I_4}{4}$ carries equal weight $\tfrac14$ on each of $\{ \Phi^{\pm}, \Psi^{\pm} \}$, hence parity is $\Phi$-type with probability $\tfrac12$. Hence the parity check
accepts with probability $\tfrac12$ regardless of the inputs. Now, the fidelity given the acceptance depends on the kept (source) pair, which is $\tfrac{I_4}{4}$, whose $\Phip$ fraction is
$\bra{\Phip}\tfrac{I_4}{4}\ket{\Phip}=\tfrac14$. On collecting both branches, the
unnormalised accepted-and-$\Phip$ weight and the total accepted weight are

\begin{equation}
\label{eq:branch-tallies}
\begin{split}
\mathcal{N} ={}&
\pg^2 N_0
+(1-\pg^2)\,\underbrace{\tfrac12}_{\text{accept}}
\cdot\underbrace{\tfrac14}_{\Phi^+\ \text{frac.}},
\\
\mathcal{D} ={}&
\pg^2 D_0
+(1-\pg^2)\,\underbrace{\tfrac12}_{\text{accept}},
\end{split}
\end{equation}

where $N_0,D_0$ are the clean-branch sums as in Eqs.~\eqref{eq:N0app}-\eqref{eq:D0app}.
Normalising by dividing through by $\pg^2$ so the clean terms appear, defines the additive
floor as

\begin{equation}
\label{eq:floor}
\begin{gathered}
m=\frac{1-\pg^2}{\pg^2}\cdot\frac12\cdot\frac14
=\frac{1-\pg^2}{8\pg^2}
\quad(\text{num.}),\\
\frac{1-\pg^2}{\pg^2}\cdot\frac12
=4m
\quad(\text{den.}),
\end{gathered}
\end{equation}
where the factor of $4=1/\bra{\Phip}\tfrac{I_4}{4}\ket{\Phip}$ is precisely the
reciprocal of the mixed-state $\Phip$ fraction that is, the denominator counts every
accepted white-noise event, the numerator only the quarter of them that read
$\Phip$. Hence

\begin{equation}
F'=\frac{N_0+m}{D_0+4m},
\label{eq:master-with-floor}
\end{equation}

which is the general map Eq.~\eqref{eq:twofid} of the main text.

Finally, the native (physical) acceptance probability is the total accepted
weight of Eq.~\eqref{eq:branch-tallies} before the $\pg^2$ rescaling,

\begin{equation}
p_{\mathrm{acc}}=\mathcal{D}=\pg^2 D_0+(1-\pg^2)\tfrac12=\pg^2\!\left(D_0+4m\right)
=\pg^2 D.
\label{eq:pacc}
\end{equation}

So the single quantity $D= D_0+4m$ serves both as the normaliser of the
fidelity in Eq.~\eqref{eq:master-with-floor} and, up to the $\pg^2$ prefactor, as
the acceptance probability. Fidelity and resource counting therefore share one
denominator, giving the heralding probabilities
$p_{\mathrm{DAEPP}}=\pg^2 D(\FP,\FS)$ and
$p_{\mathrm{acc}}^{\mathrm{BBPSSW}}(F)=\pg^2 D(F,F)$ used in
Sec.~\ref{ssec:resource} and Theorem~\ref{thm:resource}.

As mentioned above, for general map, the entire $s\!\leftrightarrow\!p$ asymmetry sits in the numerator's $\wone$
term. The asymmetry in the general map is,
\begin{align}
F'(\FS,\FP)-F'(\FP,\FS) =\frac{2\emeas(1-\emeas)}{3D}\,(\FS-\FP).
\end{align}

This is Eq.~\eqref{eq:asym}: the gap has
the sign of $\FS-\FP$ (keeping the higher-fidelity pair wins) and vanishes as
$\emeas\to1$, where the false-accept channel closes, and the slots become
interchangeable.

In summary, at $\FS=\FP$ the map reduces to D\"ur's symmetric
recurrence~\cite{dur1999quantum}; at $\FS=\FP$, $\emeas=\pg=1$ to the textbook
BBPSSW value $(10F^2-2F+1)/(8F^2-4F+5)$~\cite{bennett1996purification}. Table~\ref{tab:checks} collects these and other checks for the generalised map.

\begin{table*}[t]
\caption{Checks satisfied by the generalised map Eqs.~\eqref{eq:twofid}--\eqref{eq:Dfinal}.}
\label{tab:checks}
\centering
\small
\begin{tabular}{ll}
\toprule
Check & Result\\
\midrule
$\FS=\FP$ & D\"ur BBPSSW recurrence~\cite{dur1999quantum}\\
$\FS{=}\FP,\ \emeas{=}\pg{=}1$ & BBPSSW recurrence~\cite{bennett1996purification}\\
$\emeas=\pg=1$ & symmetric ideal recurrence Eq.~\eqref{eq:fstar-ideal}\\
$\emeas{=}\pg{=}1$ (two state bit-flip mixture)& Hu's formula Eq.~\eqref{eq:hu}\\
$\FS\to1$ & $\Fstar\to 3\FP/(2\FP{+}1)$\\
\bottomrule
\end{tabular}
\end{table*}

\section{Joint-isotropic generalization}
\label{app:joint}
The channel model of Sec.~\ref{sec:protocol} treats the polarisation and spatial
DoFs as independently twirled Werner states. Physically this is justified because
the two are degraded by distinct, uncorrelated mechanisms: polarisation by
intracore birefringence and stress-induced rotation, the spatial/path sector by
inter-core coupling and relative-phase drift between cores. Residual cross-talk can be
modelled by mixing the $16$-dimensional two-DoF Bell--Bell state toward the
maximally mixed state with weight $\lambda$ as,

\begin{equation}
\rho_{\rm in}\longmapsto(1-\lambda)\,\rho_P\otimes\rho_S+\lambda\,\frac{I_{16}}{16}.
\label{eq:joint-input}
\end{equation}

The DAEPP step is a fixed operation that does not adapt to its input, so it acts
on the mixture Eq.~\eqref{eq:joint-input} piece by piece. That is, applying it on each component and
combining the results with the same weights $1-\lambda$ and $\lambda$. For each
component we need only two quantities, the probability that the pair is kept and,
of those kept, the fraction that are the target state $\Phi^+$.

For the product part $\rho_P\otimes\rho_S$, these are exactly the quantities of
Lemma~\ref{lem:depp}: the pair is kept with probability $p_{\rm DAEPP}$ and the
kept pair has fidelity $\Fstar$. For the white-noise part,
$I_{16}/16=(I_4/4)_P\otimes(I_4/4)_S$ is a maximally mixed pair in each DoF. As in
Appendix~\ref{app:twofid}, keeping is decided by the spatial measurement, and a
maximally mixed spatial pair passes with probability $\tfrac12$; the kept
polarisation pair is likewise maximally mixed, so it is $\Phi^+$ with fraction
$\bra{\Phi^+}\tfrac{I_4}{4}\ket{\Phi^+}=\tfrac14$, independently of $\pg$.
Combining kept-fractions and good-and-kept fractions across the two components,

\begin{equation}
\Fstar(\lambda)=\frac{(1-\lambda)\,\Fstar\,p_{\rm DAEPP}+\lambda/8}
{(1-\lambda)\,p_{\rm DAEPP}+\lambda/2},
\label{eq:joint-exact}
\end{equation}

where the denominator is the total kept-fraction and the numerator is the
good-and-kept fraction. This is exact in $\lambda$, with $\Fstar(0)=\Fstar$ and
$\Fstar(1)=\tfrac14$, that is, as cross-talk grows, the output slides monotonically from
the clean value down to the white-noise floor $\tfrac14$. For small $\lambda$,

\begin{equation}
\Fstar(\lambda)=\Fstar-\lambda\Bigl(\Fstar-\tfrac14\Bigr)
\frac{1/2}{p_{\rm DAEPP}}+O(\lambda^2).
\label{eq:joint-firstorder}
\end{equation}

The perturbation only lowers $\Fstar$ (as $\Fstar>\tfrac14$ at any usable point)
and enters the chain result solely through the replacement
$\Fstar\to\Fstar(\lambda)$ in the per-segment factor $(4\Fstar-1)/3$, leaving the
structure of Theorem~\ref{thm:dahr-recursion} intact. The kept-fraction moves on
the same denominator, $p_{\rm DAEPP}(\lambda)=(1-\lambda)p_{\rm DAEPP}+\lambda/2$,
so the per-segment cost $1/p_{\rm DAEPP}(\lambda)$ shifts correspondingly. A complete $16\times16$ treatment tracking cross-DoF coherences,
relevant when $\lambda$ is not small or the joint noise is non-isotropic, is left
to future work.

\bibliography{apssamp}

@article{wehner2018quantum,
  title={Quantum internet: A vision for the road ahead},
  author={Wehner, Stephanie and Elkouss, David and Hanson, Ronald},
  journal={Science},
  volume={362},
  number={6412},
  pages={eaam9288},
  year={2018},
  publisher={American Association for the Advancement of Science}
}

@article{kumar2025quantum,
  title={Quantum internet: Technologies, protocols, and research challenges},
  author={Kumar, Vinay and Cicconetti, Claudio and Conti, Marco and Passarella, Andrea},
  journal={International Journal of Networked and Distributed Computing},
  volume={13},
  number={2},
  pages={22},
  year={2025},
  publisher={Springer}
}

@article{kumar2026making,
  title={Making Quantum Networks Work: Routing, Calibration, and Programmable Quantum Repeaters},
  author={Kumar, Vinay},
  journal={arXiv preprint arXiv:2606.22316},
  year={2026}
}

@article{kimble2008quantum,
  title={The quantum internet},
  author={Kimble, H Jeff},
  journal={Nature},
  volume={453},
  number={7198},
  pages={1023--1030},
  year={2008},
  publisher={Nature Publishing Group}
}

@article{sangouard2011quantum,
  title={Quantum repeaters based on atomic ensembles and linear optics},
  author={Sangouard, Nicolas and Simon, Christoph and De Riedmatten, Hugues and Gisin, Nicolas},
  journal={Reviews of Modern Physics},
  volume={83},
  number={1},
  pages={33--80},
  year={2011},
  publisher={APS}
}

@article{muralidharan2016optimal,
  title={Optimal architectures for long distance quantum communication},
  author={Muralidharan, Sreraman and Li, Linshu and Kim, Jungsang and L{\"u}tkenhaus, Norbert and Lukin, Mikhail D and Jiang, Liang},
  journal={Scientific reports},
  volume={6},
  number={1},
  pages={20463},
  year={2016},
  publisher={Nature Publishing Group UK London}
}

@article{zukowski1993event,
  title={‘‘Event-ready-detectors’’Bell experiment via entanglement swapping},
  author={{\.Z}ukowski, Marek and Zeilinger, Anton and Horne, Michael A and Ekert, Aarthur K},
  journal={Physical review letters},
  volume={71},
  number={26},
  pages={4287},
  year={1993},
  publisher={APS}
}

@article{zeilinger1997three,
  title={Three-particle entanglements from two entangled pairs},
  author={Zeilinger, Anton and Horne, Michael A and Weinfurter, Harald and {\.Z}ukowski, Marek},
  journal={Physical review letters},
  volume={78},
  number={16},
  pages={3031},
  year={1997},
  publisher={APS}
}

@article{bennett1996purification,
  title={Purification of noisy entanglement and faithful teleportation via noisy channels},
  author={Bennett, Charles H and Brassard, Gilles and Popescu, Sandu and Schumacher, Benjamin and Smolin, John A and Wootters, William K},
  journal={Physical review letters},
  volume={76},
  number={5},
  pages={722},
  year={1996},
  publisher={APS}
}

@article{deutsch1996quantum,
  title={Quantum privacy amplification and the security of quantum cryptography over noisy channels},
  author={Deutsch, David and Ekert, Artur and Jozsa, Richard and Macchiavello, Chiara and Popescu, Sandu and Sanpera, Anna},
  journal={Physical review letters},
  volume={77},
  number={13},
  pages={2818},
  year={1996},
  publisher={APS}
}

@article{briegel1998quantum,
  title={Quantum repeaters: the role of imperfect local operations in quantum communication},
  author={Briegel, H-J and D{\"u}r, Wolfgang and Cirac, Juan I and Zoller, Peter},
  journal={Physical Review Letters},
  volume={81},
  number={26},
  pages={5932},
  year={1998},
  publisher={APS}
}

@article{dur1999quantum,
  title={Quantum repeaters based on entanglement purification},
  author={D{\"u}r, Wolfgang and Briegel, H-J and Cirac, Juan Ignacio and Zoller, Peter},
  journal={Physical Review A},
  volume={59},
  number={1},
  pages={169},
  year={1999},
  publisher={APS}
}

@article{bennett1996mixed,
  title={Mixed-state entanglement and quantum error correction},
  author={Bennett, Charles H and DiVincenzo, David P and Smolin, John A and Wootters, William K},
  journal={Physical Review A},
  volume={54},
  number={5},
  pages={3824},
  year={1996},
  publisher={APS}
}

@article{kwiat1997hyper,
  title={Hyper-entangled states},
  author={Kwiat, Paul G},
  journal={Journal of modern optics},
  volume={44},
  number={11-12},
  pages={2173--2184},
  year={1997},
  publisher={Taylor \& Francis}
}

@article{zhao2023generation,
  title={Generation of hyperentangled state encoded in three degrees of freedom},
  author={Zhao, Peng and Yang, Meng-Ying and Zhu, Sha and Zhou, Lan and Zhong, Wei and Du, Ming-Ming and Sheng, Yu-Bo},
  journal={Science China Physics, Mechanics \& Astronomy},
  volume={66},
  number={10},
  pages={100311},
  year={2023},
  publisher={Springer}
}

@article{deng2017quantum,
  title={Quantum hyperentanglement and its applications in quantum information processing},
  author={Deng, Fu-Guo and Ren, Bao-Cang and Li, Xi-Han},
  journal={Science bulletin},
  volume={62},
  number={1},
  pages={46--68},
  year={2017},
  publisher={Elsevier}
}

@article{sheng2010complete,
  title={Complete hyperentangled-Bell-state analysis for quantum communication},
  author={Sheng, Yu-Bo and Deng, Fu-Guo and Long, Gui Lu},
  journal={Physical Review A—Atomic, Molecular, and Optical Physics},
  volume={82},
  number={3},
  pages={032318},
  year={2010},
  publisher={APS}
}

@article{sheng2010deterministic,
  title={Deterministic entanglement purification and complete nonlocal Bell-state analysis with hyperentanglement},
  author={Sheng, Yu-Bo and Deng, Fu-Guo},
  journal={Physical Review A—Atomic, Molecular, and Optical Physics},
  volume={81},
  number={3},
  pages={032307},
  year={2010},
  publisher={APS}
}

@article{sheng2010one,
  title={One-step deterministic polarization-entanglement purification using spatial entanglement},
  author={Sheng, Yu-Bo and Deng, Fu-Guo},
  journal={Physical Review A—Atomic, Molecular, and Optical Physics},
  volume={82},
  number={4},
  pages={044305},
  year={2010},
  publisher={APS}
}

@article{li2010deterministic,
  title={Deterministic polarization-entanglement purification using spatial entanglement},
  author={Li, Xi-Han},
  journal={Physical Review A—Atomic, Molecular, and Optical Physics},
  volume={82},
  number={4},
  pages={044304},
  year={2010},
  publisher={APS}
}

@article{ren2013hyperentanglement,
  title={Hyperentanglement purification and concentration assisted by diamond NV centers inside photonic crystal cavities},
  author={Ren, Bao-Cang and Deng, Fu-Guo},
  journal={Laser Physics Letters},
  volume={10},
  number={11},
  pages={115201},
  year={2013},
  publisher={IOP Publishing}
}

@article{ren2014two,
  title={Two-step hyperentanglement purification with the quantum-state-joining method},
  author={Ren, Bao-Cang and Du, Fang-Fang and Deng, Fu-Guo},
  journal={arXiv preprint arXiv:1408.0048},
  year={2014}
}

@article{wang2016hyperentanglement,
  title={Hyperentanglement purification for two-photon six-qubit quantum systems},
  author={Wang, Guan-Yu and Liu, Qian and Deng, Fu-Guo},
  journal={Physical Review A},
  volume={94},
  number={3},
  pages={032319},
  year={2016},
  publisher={APS}
}

@article{hu2021long,
  title={Long-distance entanglement purification for quantum communication},
  author={Hu, Xiao-Min and Huang, Cen-Xiao and Sheng, Yu-Bo and Zhou, Lan and Liu, Bi-Heng and Guo, Yu and Zhang, Chao and Xing, Wen-Bo and Huang, Yun-Feng and Li, Chuan-Feng and others},
  journal={Physical review letters},
  volume={126},
  number={1},
  pages={010503},
  year={2021},
  publisher={APS}
}

@article{victora2023entanglement,
  title={Entanglement purification on quantum networks},
  author={Victora, Michelle and Tserkis, Spyros and Krastanov, Stefan and de la Cerda, Alexander Sanchez and Willis, Steven and Narang, Prineha},
  journal={Physical Review Research},
  volume={5},
  number={3},
  pages={033171},
  year={2023},
  publisher={APS}
}

@inproceedings{benchasattabuse2025integrating,
  title={Integrating Entanglement Purification into All-Photonic Quantum Repeaters},
  author={Benchasattabuse, Naphan and Hajdu{\v{s}}ek, Michal and Van Meter, Rodney},
  booktitle={2025 IEEE International Conference on Quantum Computing and Engineering (QCE)},
  volume={1},
  pages={885--895},
  year={2025},
  organization={IEEE}
}

@article{mylavarapu2025teleportation,
  title={Teleportation fidelity of quantum repeater networks},
  author={Mylavarapu, Ganesh and Ghosh, Subrata and Hens, Chittaranjan and Chakrabarty, Indranil and Mitra, Subhadip},
  journal={Physical Review A},
  volume={112},
  number={3},
  pages={032618},
  year={2025},
  publisher={APS}
}

@article{miguel2023quantum,
  title={Quantum repeater for W states},
  author={Miguel-Ramiro, Jorge and Riera-S{\`a}bat, Ferran and D{\"u}r, Wolfgang},
  journal={PRX Quantum},
  volume={4},
  number={4},
  pages={040323},
  year={2023},
  publisher={APS}
}

@article{ghosal2025repeater,
  title={Repeater-based quantum communication protocol: Maximizing teleportation fidelity with minimal entanglement},
  author={Ghosal, Arkaprabha and Ghai, Jatin and Saha, Tanmay and Ghosh, Sibasish and Alimuddin, Mir},
  journal={Physical Review Letters},
  volume={134},
  number={16},
  pages={160803},
  year={2025},
  publisher={APS}
}

\end{document}